\RequirePackage[orthodox]{nag}

\documentclass[a4paper,USenglish]{lipics-v2021}

\usepackage{xspace}
\usepackage[noadjust]{cite}

\title{Complexity Thresholds for the Constrained Colored Token Swapping Problem}

\author{Davide Bilò}{Department of Information Engineering, Computer Science, and Mathematics\\University of L'Aquila, Italy}{davide.bilo@univaq.it}{https://orcid.org/0000-0003-3169-4300}{}
\author{Stefano Leucci}{Department of Information Engineering, Computer Science, and Mathematics\\University of L'Aquila, Italy}{stefano.leucci@univaq.it}{https://orcid.org/0000-0002-8848-7006}{}
\author{Andrea Martinelli}{Department of Information Engineering, Computer Science, and Mathematics\\University of L'Aquila, Italy}{andrea.martinelli1@student.univaq.it}{}{}

\authorrunning{D. Bilò, S. Leucci, and A. Martinelli}

\Copyright{Davide Bilò, Stefano Leucci, and Andrea Martinelli}



\newcommand{\pspace}{\ensuremath{\mathsf{PSPACE}}\xspace}
\newcommand{\np}{\ensuremath{\mathsf{NP}}\xspace}
\newcommand{\cts}{\textsc{CTS}\xspace}
\newcommand{\ccts}{\textsc{CCTS}\xspace}
\newcommand{\ncl}{\textsc{NCL}\xspace}
\newcommand{\pmg}{\textsc{PMG}\xspace}
\newcommand{\cctss}{\textsc{CCTS}\ensuremath{{}^\star}\xspace}
\newcommand{\fas}{\textsc{FAS}\xspace}

\newcommand{\B}{\ensuremath{\mathcal{B}}}

\supplement{An interactive instance resulting from the \pspace-completeness reduction in this paper is available at \url{https://www.isnphard.com/g/constrained-colored-token-swapping/}.}

\hideLIPIcs
\nolinenumbers

\begin{document}

\maketitle

\begin{abstract}
    Consider the following puzzle: a farmland consists of several fields, each occupied by either a farmer, a fox, a chicken, or a caterpillar. Creatures in neighboring fields can swap positions as long as the fox avoids the farmer, the chicken avoids the fox, and the caterpillar avoids the chicken. The objective is to decide whether there exists a sequence of swaps that rearranges the creatures into a desired final configuration, while avoiding any unwanted encounters.

    The above puzzle can be cast an instance of the \emph{colored token swapping} problem with $k = 4$ colors (i.e., creature types), in which only certain pairs of colors can be swapped.
    We prove that such problem is \pspace-hard even when the graph representing the farmland is planar and cubic. 
    We also show that the problem is polynomial-time solvable when at most three creature types are involved. We do so by providing a more general algorithm deciding instances with arbitrary values of $k$, as long as the set of all admissible swaps between creature types induces a \emph{spanning star}.

    Our results settle a problem explicitly left open in [Yang and Zhang, IPL 2025], which established $\mathsf{PSPACE}$-completeness for eight creature types and left the complexity status unresolved when the number of creature types is between three and seven.
\end{abstract}

\section{Introduction}

\begin{figure}[t]
    \centering
    \includegraphics[scale=0.74]{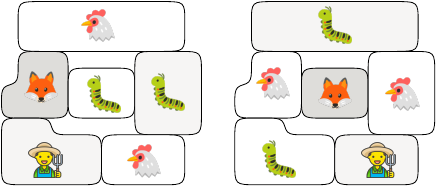}
    \caption{A possible arrangement of fields and creatures of the brain teaser.}
    \label{fig:example}
\end{figure}

Consider farmland that comprises of multiple fields, each of which is occupied by either a farmer, a fox, a chicken, or a caterpillar, as in Figure~\ref{fig:example}. 
The creatures on two neighboring fields can swap their positions, causing them to meet along the way. However, the fox does not want meet the farmer, the chicken does not want to meet the fox, and the caterpillar does not want to meet the chicken. 
The goal is repositioning the creatures so that, in their resulting arrangement, each field
is occupied by a creature of a given type, and no undesirable meetings occur in the process.

This puzzle, which is a generalization of a well-known brain teaser, can be cast as an instance of the \emph{colored token swapping} (\cts) problem on graphs \cite{MiltzowNORTU16}, with the additional constraint that only some of the swaps are allowed.
In the \cts problem we have an undirected and connected \emph{base graph} $G$, in which each vertex
contains a \emph{token} of one out of $k$ possible \emph{colors} in $\{1, 2, \dots, k\}$, where tokens of the same color are indistinguishable.
Two tokens on neighboring vertices can be \emph{swapped} and the problem asks, given an \emph{initial} and a \emph{final} arrangement of tokens on $G$, to find a sequence of swaps that transforms the initial configuration into the final one.
Since all \cts instances can be solved, the scientific community has focused on studying the optimization version of the problem, where the goal is minimize the number of swaps \cite{MiltzowNORTU16,AkersK89,vaughan1999factoring,YamanakaDIKKOSS15,HikenW25}.

In the \emph{constrained} version of \cts (\ccts for short), we are additionally given a \emph{swap graph} whose set of vertices is exactly the set of token colors $\{1, 2, \dots, k\}$  and a swap between two tokens is only allowed if their colors are adjacent in the swap graph. An example instance of \ccts is shown in Figure~\ref{fig:example2}.
A famous puzzle which is captured by \ccts is the Fifteen Puzzle \cite{wilson1974graph} since each numbered tile can be modeled by a token with a unique color, and the blank tile can be interpreted as an additional token of a sixteenth color, which is allowed to swap with any other color. 
More broadly, \ccts generalizes the \emph{pebble motion problem} studied by Kornhauser, Miller, Spirakis \cite{KornhauserMS84} which extents the Fifteen Puzzle to arbitrary graphs and multiple blank positions. In contrast with \cts, an instance of \ccts is not necessarily solvable, which motivates the study of its decision version.

The version of \ccts in which all tokens have distinct colors is known as the \emph{Friends and Strangers} (\fas) problem, which was introduced in \cite{defant2021friends}.
In their paper, the authors prove several structural properties of the \emph{friends-and-strangers graph} whose vertices represent all possible arrangements of tokens, and two vertices are adjacent if the two configuration differ by a single permissible swap. Additional structural properties of the friends-and-strangers graph have been shown in \cite{Jeong26}.
Concerning the computational complexity, \cite{YangZ25} proved that \fas is \pspace-complete even when the base graph is subcubic and planar. Actually, the instances of \fas considered in \cite{YangZ25} are equivalent to \ccts instances with $8$ colors
and the swap graph in Figure~\ref{fig:example2}~(d). In the concluding remarks of their paper, the authors leave open the problem of understanding the computational complexity of \ccts when the number of colors is between $3$ and $7$.\footnote{\ccts is decidable in polynomial time when there are only two colors, since the swap graph must be either empty or complete (in this latter case, the \ccts instance is actually an instance of \cts).}

\subparagraph*{Our results.} 
We close the problem left open in \cite{YangZ25} by proving that \ccts is \pspace-hard for $k \ge 4$, and decidable in polynomial time for $k = 3$.

In particular, the \pspace-hardness result holds whenever the swap graph contains an induced path on four vertices, as it is the case in the above brain teaser, and even if the base graph is cubic and planar.
On the positive side, we actually show a more general result: the \ccts problem is decidable in polynomial time whenever the swap graph is a star (regardless of the structure of the base graph).
This implies that the problem is solvable in polynomial time when $k=3$ since the swap graph is either disconnected, complete, or a star, and the former two cases are easy to handle.
This result extends the characterization given by Kornhauser, Miller, Spirakis \cite{KornhauserMS84} for the pebble motion problem on graphs. 

\begin{figure}
    \centering
    \includegraphics[scale=0.85]{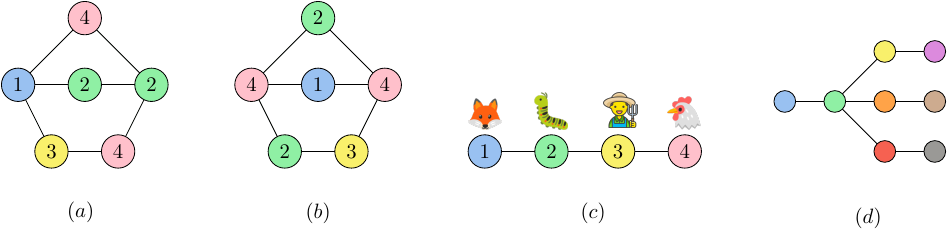}
    \caption{(a), (b), and (c): The initial configuration (a), the final configuration (b) and the swap graph (c) of the \ccts instance corresponding to the brain teaser in Figure~\ref{fig:example}. (d): The swap graph used in the \pspace-hardness reduction of \ccts given in \cite{YangZ25}.}
    \label{fig:example2}
\end{figure}

\subparagraph*{Related works.} 
We are not aware of any other work on \ccts other than those mentioned in the introduction.
Regarding \cts, the problem is known to be solvable in polynomial time when the graph is a star or a path \cite{BonnetMR18}, a clique with $O(1)$ colors \cite{YamanakaHKKOSUU18}, a complete bipartite graph with $O(1)$ colors \cite{BiloFG024}, and when $k=2$ \cite{YamanakaHKKOSUU18}. In the special case in which all token colors are distinct, the problem is also polynomial-time solvable on cliques \cite{cayley1849lxxvii}, cycles \cite{Jerrum85}, brooms \cite{vaughan1999factoring,KawaharaSY19}, lollipops \cite{KawaharaSY19}, squares of paths \cite{HeathV03}, complete split graphs \cite{yasui2015swapping}, complete bipartite graphs \cite{YamanakaHKKOSUU18}, and cographs \cite{TonettiSU24}.
On the negative side, \cts is \np-hard for planar bipartite graphs with maximum degree $3$ \cite{YamanakaHKKOSUU18}, cliques \cite{BonnetMR18}, complete bipartite graphs \cite{BiloFG024}, and trees \cite{AichholzerDKLLM22}, even when all token colors are distinct.
Finally, we mention that \cts is $\mathsf{APX}$-hard \cite{MiltzowNORTU16} and its approximability has received considerable attention (see \cite{HikenW25} and the references therein).

\subparagraph*{Preliminaries.} 

In the following we assume, w.l.o.g., that the base graph contains at least one token of each color in the swap graph, since otherwise one can equivalently consider the instance in which all unused colors (i.e., vertices) have been deleted from the swap graph.
Moreover, we assume that the swap graph is connected, since otherwise the problem reduces to solving multiple instances of \ccts, where each instance has (i) the swap graph induced by a connected component $C$ of the original swap graph, and (ii) the base graph induced by all and only the vertices of the original base graph that contain a token some color in $C$.
Finally, we assume that the number of tokens of each color in the initial and final configurations is the same, since otherwise we trivially have a no-instance.

We describe our \pspace-hardness result in the next section, our  algorithm for star swap graphs is discussed in Section~\ref{sec:poly}, and we provide some concluding remarks in Section~\ref{sec:conclusion}.

\section{PSPACE-Completeness}

In this section we establish the \pspace-completeness of \ccts even when the base graph $G$ is cubic and planar, and the swap graph is a path on $4$ vertices. We assume w.l.o.g.\ that a token of color $i$ can only be swapped with tokens of color $i-1$ (if such a color exists) or $i+1$ (if such a color exists), as in the swap graph of Figure~\ref{fig:example}~(c).

We prove our result by providing a reduction from the \emph{configuration-to-configuration nondeterministic constraint logic problem on cubic planar AND/OR constraint graphs} \cite{HearnD05}, which we abbreviate with \ncl for conciseness. 
To define the \ncl problem, consider an undirected cubic planar graph $\overline{G}$ in which each edge has a weight in $\{1,2\}$. 
We refer the edges of weight $1$ (resp.\ weight $2$) as \emph{light} edges (resp.\ heavy edges).\footnote{In \cite{HearnD05} these edges are called red and blue. We use the terms light and heavy to avoid confusion with token colors.}
Each \emph{node}\footnote{We use the term \emph{node} when referring to vertices of a \ncl graph, and the term \emph{vertex} when referring to vertices of a \ccts graph.} of $\overline{G}$ has exactly three incident edges and is either an \emph{AND} node or an \emph{OR} node: an AND node has one incident heavy edge and two incident light edges, while all three edges incident to an OR node are heavy.
A \emph{configuration} of $\overline{G}$ is an orientation of its edges, and it is said to be \emph{valid} if, for each node $v$, the sum of the weights of the in-edges of $v$ is at least $2$. 
This means that if $v$ is an OR node, then at least one of $v$'s incident edges needs to be directed towards $v$. When $v$ is an AND node, then the following implication must hold: if the heavy edge incident to $v$ is directed away from $v$ then both light edges incident to $v$ must be directed towards $v$.
In the \ncl problem we are given two valid configurations of the same graph $\overline{G}$, referred to as the \emph{initial} and the \emph{final} configuration, and the goal is deciding whether it is possible to obtain the final configuration from the initial one via a sequence of edge flips that only result in intermediate valid configurations.
This problem is known to be \pspace-complete \cite{HearnD05}.

To reduce an instance of \ncl to an instance of \ccts, we show how to replace each edge, OR node, and AND node of a configuration with a suitable \emph{gadget} subgraph. 
Each gadget for an AND or an OR node will have three distinct vertices designated as \emph{connection points}, one for each of the edges incident to the node. Each edge gadget has two vertices designated as connection points. If an edge $e$ is incident to node $v$ in the \ncl instance, then a suitable connection point of $v$'s gadget will be linked, though an undirected edge, with one of the connection points of $e$'s gadget.
Gadgets will only interact by exchanging tokens placed on their connection points, and the only tokens to cross gadget boundaries will have colors $1$ or $2$, while tokes of color $3$ and $4$ will be confined within their starting gadgets. It is useful to think of tokens of color $1$ as ``fillers'' and to focus on tokens of color $2$ instead. This is because the only way to flip an edge requires bringing a token $t$ of color $2$ into that edge's gadget and, as a result, $t$ moves from the node gadget of the old edge head to the node gadget of the new edge head.

In the following subsections we first describe the gadgets, and then we argue on the correctness of our reduction. For the sake of simplicity, we first provide a reduction where the base graph  of \ccts is planar with maximum degree $3$, and then we modify the reduction to obtain a cubic base graph while still preserving planarity. 

\subsection{The edge gadget} 

\begin{figure}[t]
    \centering
    \includegraphics[scale=0.76]{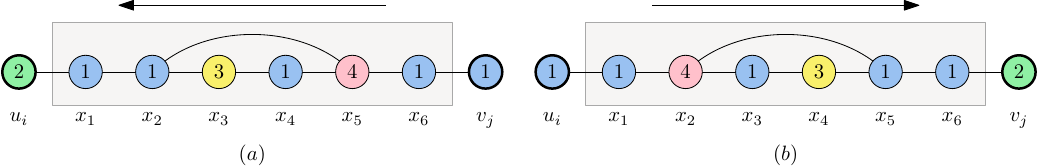}
    \caption{An edge gadget oriented towards $u$ (a) and towards $v$ (b), enclosed in a gray rectangle. The connection points of the edge gadget are vertices $x_1$ and $x_6$. The connection points $u_i$ and $v_j$ of the gadgets modeling vertices $u$ and $v$ are shown in bold.}
    \label{fig:edge_gadget}
\end{figure}

The edge gadget is used to model an edge $e = (u,v)$ of an \ncl configuration, and is the same for both light and heavy edges.
It consists of a path with $6$ vertices $x_1, x_2, \dots, x_6$ (in order from one endvertex to the other) plus the additional edge $(x_2, x_4)$, is connected to the gadget representing node $u$ via the edge $(u_i, x_1)$, where $u_i$ is a connection point that belongs to the gadget of $u$, and to the gadget representing node $v$ via the edge $(x_6, v_j)$, where $v_j$ is a connection point that belongs to the gadget of $v$.

The idea of the gadget is the following: an arrangement of tokens always defines an orientation, which can be either towards $u$ (and away from $v$) or towards $v$ (and away from $u$). The gadget's orientation matches that of edge $e$ in \ncl configuration. 
The gadget cannot change orientation on its own, but this becomes possible if it receives a token of color $2$ from the (gadget of the) node it is directed towards. This enables a sequence of token swaps that flips the gadget's orientation while moving the token into the (gadget of the) opposite node.

When the gadget is not undergoing a change in orientation, it is in one of the two  configurations shown in Figure~\ref{fig:edge_gadget}, which are mirrored versions of each other.
The sequences of token colors, when read from $x_1$ to $x_6$, is either $\langle 1,1,3,1,4,1 \rangle$ or $\langle 1,4,1,3,1,1 \rangle$. The former corresponds to orienting $e$ towards $u$, while the latter corresponds to orienting $e$ towards $v$.
In general, all reachable configurations confine the token of color $4$ to be either on $x_2$ or on $x_5$. In the former case we say that the gadget points towards $v$ (and away from $u$), while in the later the gadget points towards $u$ (and away from $v$).

Suppose that the gadget is currently oriented towards $u$, that a token of color $2$ is placed on $u_i$, and that a token of color $1$ is placed on $v_j$, as in Figure~\ref{fig:edge_gadget}~(a). Then it is possible to reverse its orientation via the sequence of swaps shown (compactly) in Figure~\ref{fig:edge_gadget_flip}. In this sequence, the token of color $2$ \emph{traverses} the gadget from left to right, enabling a swap between the tokens with colors $3$ and $4$, and ends up on vertex $v_j$. Observe that a token of color $1$ is placed on $u_i$ as a result of the flip.

\begin{figure}[t]
    \centering
    \includegraphics[scale=0.76]{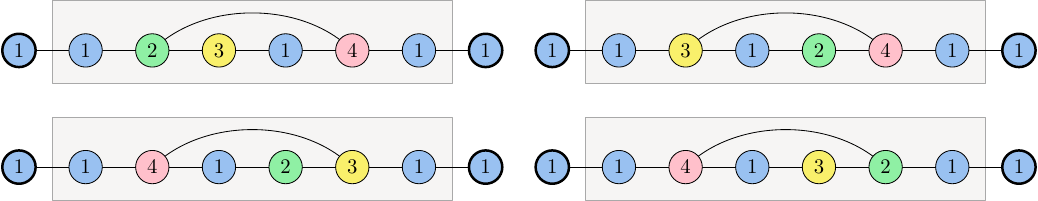}
    \caption{From top left to bottom right: four intermediate configuration encountered when the of color $2$ in Figure~\ref{fig:edge_gadget}~(a) is used to traverse the gadget, causing it to flip its orientation and to end up in the configuration shown in Figure~\ref{fig:edge_gadget}~(b). For the sake of conciseness, not all swaps are shown.}
    \label{fig:edge_gadget_flip}
\end{figure}

An important observation is that, while it is possible for a token of color $2$ to enter the gadget from $v$'s side, such token cannot be used to traverse the gadget as it is blocked by the token of color $4$ placed in $x_5$. This holds even if multiple tokens of color $2$ enter the gadget from $v$'s side.
Moreover, the gadget behaves as intended even if two tokens of color $2$ try to simultaneously traverse the gadget from $u$'s side to $v$'s side. Indeed, only one of these tokens can get to $v$'s side of the gadget, and this is only possible by flipping the gadget's orientation. In any case, the other token remains confined on $u$'s side, in either $x_1$ or $x_2$.

A symmetric argument holds when the gadget is directed towards $v$: a token of color $2$ entering from $v$'s side can reverse the gadget's orientation, and the gadget behaves correctly when two such tokens enter from $u$, or multiple tokens of color $2$ enter from $u$'s side.

\subsection{The OR gadget}
\label{sub:or_gadget}

The OR gadget is used to model each OR node $v$ of an \ncl configuration, and it consists of a clique on $3$ vertices $v_1, v_2, v_3$, each of which which serves as a connection points for the edge gadgets of the three edges of $\overline{G}$ incident to $v$. 
The vertices of the gadgets will only contain tokens of color $1$ or $2$, and there will be at most $2$ tokens of color $2$. The specific placement of these tokens in the vertices of the gadget will not be relevant.
In all configurations of the gadget, the number of edges incident to $v$ that are oriented towards $v$ will be at least $\eta + 1$, where $\eta$ is the number of tokens of color $2$ in the gadget.
This ensures that at least one of the incident edges must be oriented towards $v$ at all times (see Figure~\ref{fig:or_gadget}).

\begin{figure}
    \centering
    \includegraphics[scale=0.76]{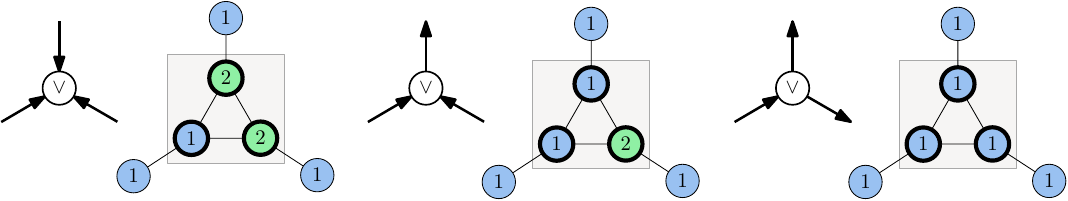}
    \caption{Three different orientation of the edges incident to an OR node, each  with a corresponding configuration of the tokens in the associated OR gadget (enclosed in a gray rectangle). All vertices of the OR gadget are connection points and are shown in bold. The neighboring connection points of the edge gadgets are shown outside the gray rectangles. Only the number of tokens of color $2$ in the gadget is relevant, not the specific vertices on which they are placed.}
    \label{fig:or_gadget}
\end{figure}

If $\eta \ge 1$ and some edge gadget connected to a generic vertex $v_i$ is directed inward, we can flip the edge gadget by first placing a token color $2$ in $v_i$, which requires at most one swap with some vertex $v_j$ of the OR gadget, and then we use such token to traverse the edge gadget. 
Conversely, to flip an edge gadget that is currently oriented outward and is connected to vertex $v_i$, we first ensure that $v_i$ contains a token of color $1$ (which can always be guaranteed by performing at most one swap with some vertex $v_j$), and then we perform a traversal of the edge gadget, which moves a token of color $2$ into $v_i$. 

Observe that it is possible for two tokens of color $2$ to leave the gadget through the same connection point and enter the same edge gadget, however at most one of then will be able to traverse the edge gadget (by flipping it, if its current orientation allows), while the other will be blocked.

\subsection{The AND gadget}

The AND gadget is used to model each AND node $v$ of a \ncl configuration. 
It consists of a path on $6$ vertices  $v_1, v_2, \dots, v_6$ (in order from one endvertex to the other) where $v_2$, and $v_4$ serve as connection points for the (edge gadgets of modeling the) two light edges incident to $v$, while $v_5$ is the connection point for the (edge gadget modeling the) heavy edge incident to $v$ (see Figure~\ref{fig:and_gadget}).

There are two main (classes of) arrangements of the tokens in the gadget, which we will refer to as the \emph{up state} and \emph{down state}. 
In the up state, the sequence of colors of the token on vertices $v_1, \dots, v_4$ is $\langle 2, 4, 2, 4 \rangle$, $v_6$ has a token of color $3$, and $v_5$ has a token of a color in $\{1, 2\}$. Examples are shown in Figure~\ref{fig:and_gadget}~(a) and (b).
In the down state $v_1$ contains a token of color $3$, $v_3$ and $v_5$ contain tokens of color $4$, $v_6$ contains a token of color $2$, and the tokens in $v_3$ and $v_4$ have a color in $\{1,2\}$. The two possible arrangements are shown in Figure~\ref{fig:and_gadget}~(c) and (d).

The behavior of the gadget is the following: in the down state, the (edge gadget corresponding to the) heavy edge incident to $v$ must be directed towards $v$ and it cannot be flipped, since it is impossible for a token of color $2$ to enter the edge gadget.
Indeed, any such token is blocked by token of color $4$ in vertex $v_5$. The two light edges can point in any direction and they can be freely reoriented, as long as the gadget of their other endpoint allows. When a light edge points towards (resp.\ away from) $v$, a token of color $2$ can be placed (resp.\ a token of color $1$ is placed) on that edge's connection point.
In the up state, the light edges incident to $v$ must be directed towards $v$, and cannot be flipped (since the only tokens of color $2$ that can reach the corresponding edge gadgets are ``locked'' in vertices $v_1$ and $v_3$ by the tokens of color $4$ in vertices $v_2$ and $v_4$). The heavy edge incident to $v$ can be directed either towards or away from $v$, and it can be flipped by traversing it with the token of color $2$ which either placed in or removed from vertex $v_5$.

\begin{figure}
    \centering
    \includegraphics[scale=0.76]{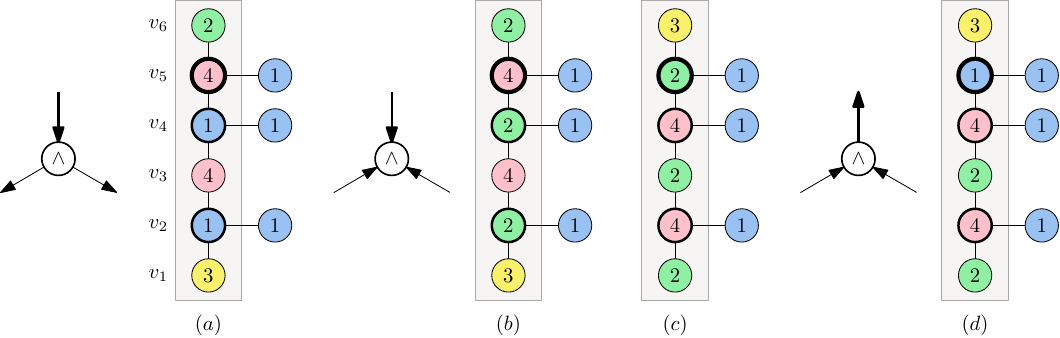}
    \caption{Several possible configurations of the AND gadget (enclosed in a gray rectangle) along with the corresponding AND node of the \ncl instance and its incident edges (shown on the left of the gadget). The connection points of the gadget are in bold, and the line thickness indicates whether the corresponding edge is light or heavy.  The neighboring connection points of the edge gadgets are shown outside the gray rectangles.
    (a) and (b) show two possible arrangement of the tokens when the AND gadget is in the down state. (c) and (d) show two arrangements in the up state. Notice how it is possible to transition from (b) to (c) by ``sliding'' the token of color $3$ from $v_1$ to $v_6$, and vice versa.}
    \label{fig:and_gadget}
\end{figure}

It is possible to transition from the down state to the up state by ``sliding'' the token of color $3$ from vertex $v_1$ to vertex $v_6$ following the path in the gadget. Since a token of color $3$ can only be swapped with tokens of colors $2$ or $4$, this requires vertices $v_2$ and $v_4$ to contain tokens of color $2$, which forces the (edge gadgets corresponding to the) light edges to be directed towards $v$. This also places a token of color $2$ on $v_5$.
Similarly, it is possible to transition from the up state to the down state by ``sliding'' the token of color $3$ from vertex $v_6$ to vertex $v_1$. This locks a token of color $2$, which must be initially placed in $v_5$, into vertex $v_6$, thus forcing the heavy edge to be directed towards $v$. Moreover, two other tokens of color $2$ are moved from $v_1$ and $v_3$ to the connection points for the light edges.

Notice that, while it is possible for the gadget to be in an intermediate state which is neither up nor down, it is never convenient to leave the gadget with such an arrangement of tokens. Indeed, any such configuration still has a token of color $2$ in vertex $v_6$, and it also prevents flipping of either one or both of the light edges, i.e., it is more restrictive than the down state.

\subsection{Combining the gadgets}

\begin{figure}[t]
    \centering
    \includegraphics[scale=0.76]{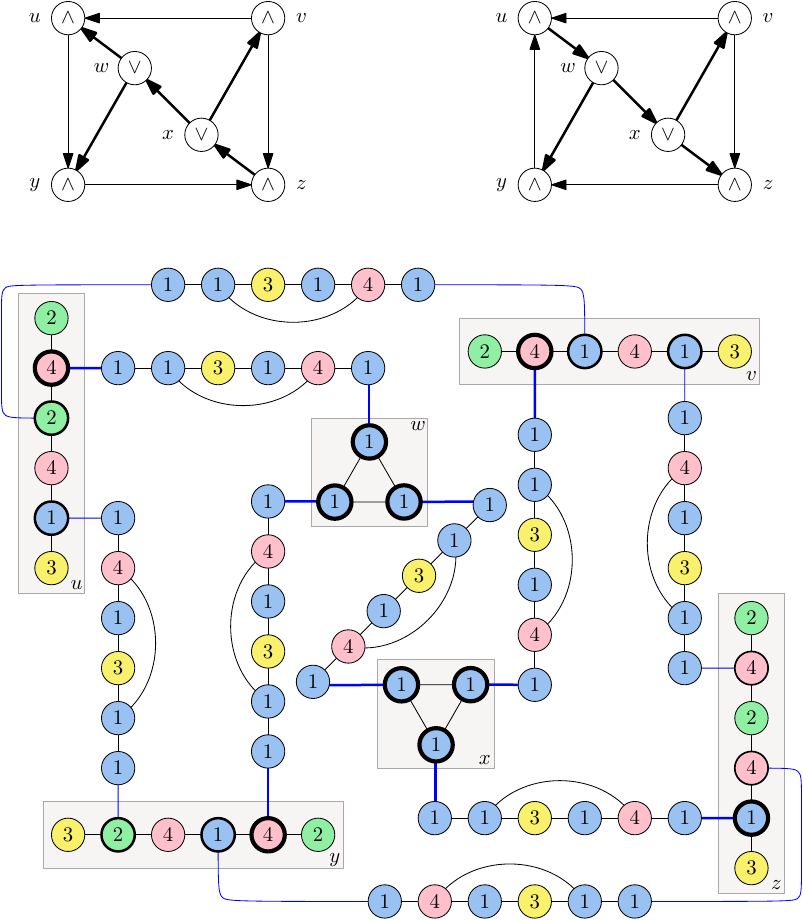}
    \caption{Top left and top right: the initial and final configurations of a \ncl instance, respectively. Heavy edges are bolder.
    Bottom: the initial configuration of the \ccts instances obtained by applying our reduction the the initial \ncl configuration. AND and OR gadgets are enclosed in gray rectangles labeled with the name of the corresponding node in the \ncl graph, while edges linking a the connection points of AND/OR gadgets with those of edge gadgets are shown in blue, and their thickness indicates whether the corresponding edge of the \ncl instance is light or heavy.}
    \label{fig:instance}
\end{figure}

The initial and the final configurations of the \ccts instance are respectively obtained by replacing each edge, AND vertex, and OR vertex of the initial and final configurations for the \ncl instance with the respective gadgets described above, where the appropriate arrangement of tokens is chosen to reflecting the edge orientations in \ncl instance. Figure~\ref{fig:instance} shows an example of a \ncl instance, and the \ccts configuration resulting by applying our reduction to the initial \ncl configuration. An interactive version of Figure~\ref{fig:instance} is available at \url{https://www.isnphard.com/g/constrained-colored-token-swapping/}.

Since the graph of the \ncl instance is planar, and each gadget is outerplanar, the resulting graph is also planar. Moreover, it is immediate to check that each vertex of the \ccts instance has degree at most $3$.
We now show that our \ccts instance is solvable if and only if the \ncl instance is solvable, thus establishing the \pspace-completeness of \ccts with $4$ colors.\footnote{The membership of \ccts in \pspace has been shown in \cite{YangZ25}.}

\subparagraph{A solution for \ncl implies a solution for \ccts.}
Suppose that the \ncl instance is solvable, and let $\langle e_1, e_2, \dots, e_k \rangle$ be a sequence of edges to be flipped to transforms the initial configuration into the final configuration while ensuring that all intermediate configurations remain valid. Let $G_\ell$ be the directed graph obtained after performing the first $\ell$ edge flips, so that $G_0$ (resp.\ $G_k$) coincides with the initial (resp. final) \ncl configuration.

We can solve the \ccts instance by simulating each edge flip, in the same order. More precisely, immediately before the generic $\ell$-th edge flip, the \ccts configuration will be exactly the one that would be obtained by applying our reduction directly from $G_{\ell-1}$, and we show a sequence of token swaps transforming such configuration into that obtained by applying our reduction to $G_k$.
To describe how the generic $\ell$-th edge flip is implemented, let $u$ and $v$ be the head and the tail nodes of $e_\ell$ in $G_{\ell-1}$, respectively, so that $e_\ell$ needs to be reoriented from pointing towards $u$ to pointing towards $v$. Let $u_i$ (resp.\ $v_j$) be the connection point of $u$'s gadget (resp.\ $v$'s gadget)
which is incident to (a connection point of the edge gadget modeling) $e_\ell$. 
We proceed in multiple steps, each of which can consist of zero, one, or multiple token swaps.

The first step takes care of placing a token of color $2$ initially sitting on some vertex of $u$'s gadget on vertex $u_i$, if that is not already the case.
If $u$ is an OR node, then $e_\ell$ is a light edge and $u$ has at least one other incoming edge in $G_{\ell-1}$ other than $e_\ell$. Then, $u$'s gadget contains at least one token of color $2$ which is either already in $u_i$ or can be moved to $u_i$ with one swap. 
Otherwise, $u$ is an AND node and we proceed differently depending on whether $e_\ell$ is a light or a heavy edge. If $e_\ell$ is a light edge, then $u$ has an incoming heavy edge in $G_{\ell-1}$, and $u$'s gadget is an AND gadget in the down state. Then, vertex $x$ already contains a token of color $2$ and no swap is necessary.
If $e_\ell$ is heavy edge, then $u$ has two incoming light edges in $G_{\ell-1}$, $u_i = u_5$, and $u$'s gadget is in the down state with two tokens of color $2$ on $u_2$ and $u_4$. We slide the token of color $3$ is $u$'s gadget to transition it to the up state, thus placing a token of color $2$ on $u_5$.

The second step flips the direction of the edge gadget modeling $e_\ell$ from pointing towards $u$ to pointing towards $v$. This is done by using the token $t$ of color $2$, which is now on $u$'s connection point for $e_\ell$, to traverse  $e_\ell$'s gadget. As a result, $t$ is moved from $u_i$ to $v$'s connection point for $e_\ell$, i.e., to $v_j$.

The third and final step may rearrange the tokens in $v$'s gadget, and is only needed when $v$ is an AND node and $e_\ell$ is a heavy edge.
In this case, $e_\ell$ is oriented away from $v$ in $G_{\ell-1}$, which implies that both light edges incident to $v$ are oriented inward. Then, $v$'s gadget is in the up state and, after the second step, vertex $v_5$ contains a token of color $2$. We transition the gadget to the down state by sliding the token of color $3$ from $v_6$ to $v_1$.

\subparagraph{A solution for \ccts implies a solution for \ncl.}
Suppose now that the \ccts instance is solvable. We show that this implies the existence of a sequence of edge flips solving the \ncl instance. In particular, we look at a sequence of swaps solving the \ccts instance and, whenever the tokens of color $3$ and $4$ are swapped in an edge gadget, we swap the corresponding edge in the \ncl instance to match the new orientation of the gadget. We claim that all \ncl configurations encountered during such a sequence of edge flips are valid. Since the direction of the edge gadgets always match the directions of the associated \ncl edges, this implies that the configuration obtained after performing all edge flips is exactly the sought final configuration of the \ncl instance. 

Given any \ccts configuration, a node $w$, and an edge $e$ incident to $w$, we denote by $w_e$ the connection point of $w$'s gadget
that is connected to a connection point of $e$'s gadget.
We define the \emph{extended version} of $w_e$ as the set of vertices containing $w_e$ and all the vertices of $e$'s gadget that are reachable from $w_e$ without traversing vertices with a token of color $4$.

Consider a generic swap involving the tokens of color $3$ and $4$ in an edge gadget representing some edge $e = (u, v)$, and assume (w.l.o.g.) that $e$'s gadget flips from being directed towards $u$ before the swap to being directed towards $v$ after the swap.
Changing the direction of $e$'s gadget requires interacting with token of color $2$, which is moved from the extended version of $u_e$ to the extended version of $v_e$ as a result of the flip.

If $u$ is an OR node, then $u$'s gadget had at least two incident edge gadgets directed towards $u$ before the swap (see Section~\ref{sub:or_gadget}), which implies that flipping $e$ still results in a \ncl configuration in which $u$ has an incoming edge, showing that the configuration is valid.

If $u$ is an AND node, we distinguish two cases depending on whether $e$ is light or heavy.
If $e$ is light, then $u_e$ is either $u_2$ or $u_4$, and the flip of $e$'s gadget must have been enabled by some token of color $2$ initially into the extended version of $u_e$.
This implies that, immediately after the flip, $u$'s gadget cannot be in the up state, and hence it contains a token of color $2$ in $u_6$. This is only possible if the edge gadget connected to $v_5$, which represents a heavy edge, is directed towards $u$. Then, in the resulting \ncl configuration, node $u$ has an incoming heavy edge, hence such configuration is valid.
If $e$ is heavy, then $u_e$ is $u_5$ and the only way to flip $e$'s gadget from pointing towards $u$ to pointing towards $v$ is for a token of color $2$ to be moved from $u_e$'s extended version to $v_e$'s extended version.
This can only happen if the token of color $3$ is $u$'s gadget is placed in $v_6$, i.e., if $u$'s gadget is in the up state. This implies that the edge gadgets of the two light edges incident to $u$ are directed towards $u$. Then, after the swap, node $u$ has two incoming light edges, showing that the resulting \ncl configuration is valid.

\subparagraph{Extending the result to cubic planar base graphs.}
The arguments above show the \pspace-completeness of \ccts when the swap graph is a path on $4$ vertices and the base graph is planar and subcubic. We can extend the above result to cubic graphs by ``padding'' the edge gadget and the AND gadgets with dummy vertices and edges, which do not alter the operation of the gadgets but ensure that all the vertex degrees are exactly $3$.\footnote{The OR gadgets require no special care, since all their vertices already have degree $3$.}
Two configurations of dummy vertices allowing to increase the degree of a vertex from $1$ to $3$ and from $2$ to $3$ are shown in Figure~\ref{fig:cubic}~(a) and (b), respectively.
Figure~\ref{fig:cubic}~(c) and (d) shows the padded versions of the edge gadget and the AND gadget, respectively. Notice that, while it is possible for a token of color $2$ to swap with the additional dummy vertices of color $1$, this is useless as it does not enable any additional swap involving non-dummy vertices, and a token of color $2$ in a dummy vertex can only return to the same non-dummy vertex it left.
While the padded gadgets are no longer outerplanar, they are still planar and all connection points can be made incident to the outer face of an embedding (as it is the case for the embeddings of Figure~\ref{fig:cubic}), hence the planarity of the resulting base graph is preserved.

\begin{figure}[t]
    \centering
    \includegraphics[scale=0.76]{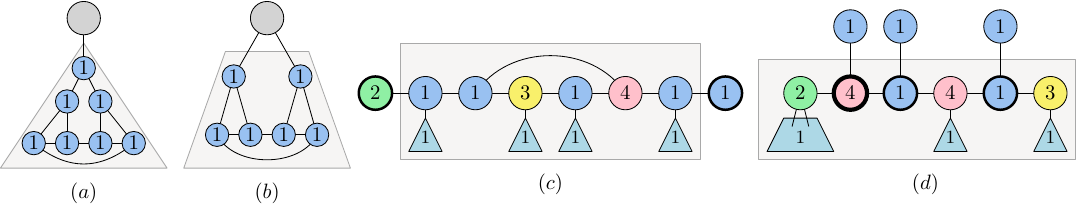}
    \caption{(a) and (b): Two configurations of dummy vertices having degree $3$ and containing tokens of color $1$ which allow to increase the degree of a vertex $v$ (in gray) from either $1$ (in (a)) or $2$ (in (b)) to $3$. (c) and (d): Padded versions of the edge gadget in the configuration shown in Figure~\ref{fig:edge_gadget}~(a), and of the AND gadget in the configuration of Figure~\ref{fig:and_gadget}~(a), respectively. Here, a triangle represents a padding with dummy vertices as arranged shown in (a), while a trapezoid represents the padding as shown in (b).}
    \label{fig:cubic}
\end{figure}

\noindent We have have established the main result of this section, which is stated in the following:

\begin{theorem}
    The \ccts problem is \pspace-complete even when restricted to swap graphs that are paths on $4$ vertices and to cubic planar base graphs.
\end{theorem}

\section{A Polynomial-Time Algorithm for Star Swap Graphs}
\label{sec:poly}

In this section we describe a polynomial-time algorithm for solving the \ccts problem when the swap graph is a star, which we denote by \cctss for short. This also shows that \cctss is polynomial-time solvable when the number of colors is $3$ since the only connected non-complete swap graph on three vertices is a star. W.l.o.g., we adopt the following assumptions: (i) the color of the center of the swap graph is $1$, and (ii) the tokens of color $1$ occupy the same set of vertices in both the initial and the final configuration of \cctss. Indeed, if (ii) is not immediately satisfied, one can always move the tokens of color $1$ to their final positions, disregarding the positions of the other tokens, and consider the instance that asks to transform the resulting configuration into the sought final configuration. Since swaps are reversible, this cannot turn a yes-instance into a no-instance, and vice versa.

Observe that a token of color $1$ can swap with a token of any other color, and that no other swaps are possible.
Hence, such an instance can be equivalently thought of as a \emph{sliding token} problem on graphs, where vertices occupied by tokens of color $1$ are considered to be \emph{blank}, i.e., devoid of tokens.
Then, the only surviving tokens have colors other than $1$, and a move consists in sliding any such token along an edge from its current vertex to neighboring blank vertex. The problem asks to decide whether, given an initial and a final configuration, the latter can be reached from the former via sequence of moves.
When all tokens of such an instance have different colors, this is exactly the \emph{Pebble Motion on Graphs} (\pmg) problem, which is a graph generalization of the 15-puzzle problem with one or more empty positions, and has been studied by Kornhauser, Miller, and Spirakis in \cite{KornhauserMS84}.\footnote{Many technical details and proofs are omitted from \cite{KornhauserMS84} and can be found in the Master's thesis by Kornhauser \cite{kornhauser1984coordinating}, on which \cite{KornhauserMS84} is based.}

In particular, \cite{KornhauserMS84} provided a characterization of the solvable instances of \pmg, which is checkable in polynomial-time. Clearly, one can check whether an \cctss instance is solvable by testing all possible \emph{color-preserving} bijections that map each token of color $c \neq 1$ in the initial configuration to a token of color $c$ in the final configuration. Each such color-preserving bijection corresponds to an instance of \pmg, whose solvability is tested using the result from \cite{KornhauserMS84}. 
However, this naive approach may require testing a superexponential number of bijections in the worst case. In the following we extend the technique employed by \cite{KornhauserMS84} for \pmg to solve \cctss in polynomial time.

We reuse a key ingredient by \cite{KornhauserMS84}, which is an equivalence relation $\sim$ that allows us to partition the vertices occupied by tokens in a \cctss instance into equivalence classes. These equivalence classes depend only on 
which set of vertices contain tokens of color other than $1$, but not on the specific token colors.
Let $\B \subseteq V(G)$ be the set of blank vertices in a graph $G = (V, E)$, and define the relation $\sim$ on the vertices not in $\B$ as follows: for $u,v \in V \setminus \B$,  vertex $u$ is related to $v$, i.e., $u \sim v$, iff there exists a sequence of swaps that brings the token initially on $u$ to vertex $v$ without changing the resulting set of blank vertices.
Moreover, for a vertex $u \in V \setminus \B$, we say that a vertex $v \in V$ is \emph{reachable} by $u$ if there exists a sequence of swaps that moves the token initially on $u$ to vertex $v$ (without necessarily preserving the set of blank vertices in the resulting configuration). We denote by $R(u)$ the set of all vertices $v$ reachable by $u$ and, for a set of vertices $U \subseteq V \setminus \B$, we let $R(U) = \bigcup_{u \in U} R(u)$.

\begin{lemma}
 $\sim$ is an equivalence relation on $V \setminus \B$.
\end{lemma}
\begin{proof}
    Clearly, $u \sim u$ for all $u \in V \setminus \B$ as it suffices to consider the empty sequence of swaps, hence $\sim$ is reflexive.

    To see that $\sim$ is symmetric, i.e., that $u \sim v$ implies $v \sim u$, let $A$ be the initial arrangement of tokens, 
    consider a sequence $S$ of swaps that places the token in $u$ into $v$, while preserving the set of blank vertices, and let $B$ be the resulting arrangement of tokens. In particular, we can think of $S$ as an ordered list of edges (possibly with repetitions) where each edge describes the swap of the two tokens residing on its endvertices, where exactly one of the two endvertices is blank.
    Let $S^{-1}$ be the sequence obtained by reversing the order of the edges in $S$. 
    If we were to apply $S^{-1}$ to $B$, this would result in configuration $A$. Hence, applying $S^{-1}$ from configuration $A$ (which is possible since the set of blank vertices is the same), moves the token on $v$ to vertex $u$ preserving the set of blank vertices.

    Finally, we show that $\sim$ is transitive. If $u \sim v$ and $v \sim w$, then let $S_1$ (resp.\ $S_2$) be a sequence of swaps that moves the token in $u$ (resp.\ $v$) to vertex $v$ (resp.\ $w$) without changing the set of blank vertices. The sequence obtained by concatenating $S_1$ with $S_2$ places the token initially in $u$ into $w$ without changing the set of blank vertices.
\end{proof}

Let $\mathcal{C} = (V \setminus \B) / \sim$ be the set of equivalence classes of $\sim$.
For each class $C \in \mathcal{C}$, we consider the subgraph $G_C = (V_C, E_C)$ of $G$ induced by the vertices in $C \cup R(C)$, and we define a sub-instance $\mathcal{I}_C$ of $\cctss$ in which the swap graph is the same. 
In the initial (resp.\ final) configuration, the vertices in $C$ contain a token with the same color as in the initial (resp.\ final) configuration of the \cctss instance.  Vertices in $R(C) \setminus C$ contain tokens of color $1$ in both the initial and in the final configuration. See Figure~\ref{fig:transitive} for an example.
Observe that it is possible for a resulting instance to have a different number of tokens of the same color in the initial and final configurations, in which case it is trivially a no-instance.

\begin{figure}
    \centering
    \includegraphics{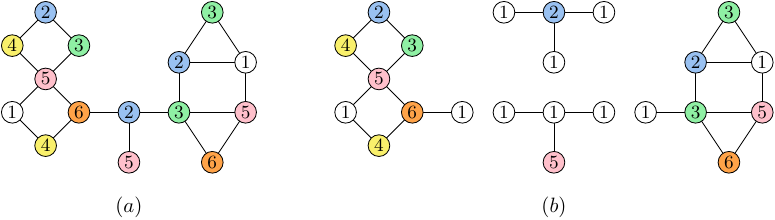}
    \caption{(a) The initial configuration of an instance  $\mathcal{I}$ of \cctss (in which the center of the swap graph is color $1$). (b) The initial configurations of the transitive sub-instances $\mathcal{I}_C$ of $\mathcal{I}$.}
    \label{fig:transitive}
\end{figure}

We say that an instance of \cctss is \emph{transitive} if it has only one equivalence class w.r.t.\ $\sim$. 
\cite{KornhauserMS84} showed how to compute both $\mathcal{C}$ and all $R(C)$ for $C \in \mathcal{C}$ in polynomial time, and showed that instances in which the vertices in $C$ are not blanks, while those in $R(C) \setminus C$ are blanks, are transitive. Then, we have the following: 
\begin{lemma}[\cite{KornhauserMS84}]
    For all equivalence classes $C \in \mathcal{C}$, the instance $I_C$ is transitive.
\end{lemma}

As discussed above, whenever a \cctss instance $\mathcal{I}$ has exactly one token for each color other than $1$, $\mathcal{I}$ is equivalent to a \pmg instance. As a consequence, its transitive sub-instances $\mathcal{I}_C$ of $\mathcal{I}$ are exactly those computed by \cite{KornhauserMS84} for \pmg, from which we state the following result, adapted to our nomenclature.

\begin{lemma}[\cite{KornhauserMS84}, rephrased]
    \label{lemma:ptg_transitive}
    Consider a transitive instance $\mathcal{I}$ of \cctss with exactly one token of each color other than $1$, and let $\mathcal{C}$ be the set of its equivalence classes w.r.t.\ $\sim$. Instance $\mathcal{I}$ is solvable if and only if all instances $I_C$, for $C \in \mathcal{C}$, are solvable.
\end{lemma}

\noindent We can show that a similar result applies to all instances of \ccts.

\begin{lemma}
 Consider a transitive instance $\mathcal{I}$ of \cctss and let $\mathcal{C}$ be the set of its equivalence classes w.r.t.\ $\sim$. Instance $\mathcal{I}$ is solvable if and only if all sub-instances $I_C$, for $C \in \mathcal{C}$, are solvable.
\end{lemma}
\begin{proof}
    Suppose that $\mathcal{I}$ is solvable, consider a sequence $S$ of swaps solving $\mathcal{I}$, and define the color-preserving bijection $\mu$ induced by $S$.
    More precisely, if the sequence $S$ of swaps moves a token from vertex $u$ in the initial configuration to vertex $v$ in the final configuration, then $\mu$ maps the token in $u$ in the initial configuration to that in $v$ in the final configuration. 
    The bijection $\mu$ induces an instance $\mathcal{I}'$ with exactly one token of each color other than $1$, which is solvable by sequence $S$. Then, Lemma~\ref{lemma:ptg_transitive} implies that all transitive sub-instances $\mathcal{I}'_C$ of $\mathcal{I}'$ are solvable.
    Then, since for each class $C \in \mathcal{C}$, the instance $\mathcal{I}'_C$ is the one induced by a color-preserving bijection of $\mathcal{I}_C$, the solvability of $\mathcal{I}'_C$ implies that $\mathcal{I}_C$ is also solvable.

    Suppose now that all sub-instances $\mathcal{I}_C$ for $C \in \mathcal{C}$ are solvable, let $\mu_C$ be the color-preserving bijection induced by a solution to $\mathcal{I}_C$, and let $\mathcal{I}'_C$ the corresponding \cctss instance with exactly one token of each color other than $1$ (which is solvable).
    Consider the function $\mu$ that maps a generic token $t$, initially sitting on a vertex belonging to equivalence class $C$, to $\mu_C(t)$.
    Since $\mu_C(t)$ must be placed on a vertex in $C$, the function $\mu$ is a bijection, which is color-preserving by construction.  

    Consider the instance $\mathcal{I}'$ induced by $\mu$. By construction, the instances $\mathcal{I}'_C$ are exactly the transitive sub-instances of $\mathcal{I}'$.
    Therefore, by Lemma~\ref{lemma:ptg_transitive}, $\mathcal{I}'$ is solvable, which implies that $\mathcal{I}$ is also solvable.
\end{proof}

\subsection{Solving transitive instances}

\begin{figure}
    \centering
    \includegraphics{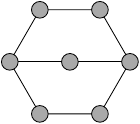}
    \caption{The graph $T_0$ mentioned in the statements of Theorem~\ref{thm:generalized_15_puzzle} and Lemma~\ref{lemma:ccts_p}.}
    \label{fig:t0}
\end{figure}

In \cite{KornhauserMS84}, the authors provide a characterization of the solvable transitive instances of \pmg, which can be checked in polynomial time. We rephrase such characterization \cctss instances where there is exactly one token of each color other than $1$.
\begin{theorem}[\cite{KornhauserMS84}, rephrased]
    \label{thm:generalized_15_puzzle}
    Consider a transitive instance $\mathcal{I}$ of \cctss with exactly one token of each color other than $1$ having a base graph $G$ that is neither a cycle nor the graph $T_0$ in Figure~\ref{fig:t0}.
    Such an instance is solvable in polynomial time and it is a no-instance if and only if
    all of the following conditions hold: (i) there is exactly one token of color $1$, (ii) $G$ is biconnected and bipartite, and (iii) the permutation induced by the initial and final positions of the tokens is odd.\footnote{Each permutation $\pi$ can be written as a product of transpositions, i.e., exchanges of two distinct elements. More precisely, each permutation $\pi$ can be written as either a product of an even number of transpositions, or as a product of an odd number of transpositions, but not both, and the \emph{parity} of $\pi$ is defined as the parity of the number of such transpositions.}
\end{theorem}

Using the above characterization, we can state the following criterion for deciding transitive \cctss instances when multiple tokens may have the same color.
\begin{lemma}
    \label{lemma:ccts_p}
    Consider a transitive instance $\mathcal{I}$ of \cctss with a base graph $G$ that is neither a cycle nor the graph $T_0$ in Figure~\ref{fig:t0}.
    Such an instance is solvable in polynomial time and it is a no-instance if and only if
    all of the following conditions hold: (i) all tokens have distinct colors, (ii) $G$ is biconnected and bipartite, and (iii) the permutation induced by the initial and final positions of the tokens (which is well-defined due to (i)) is odd.
\end{lemma}
\begin{proof}
    By Theorem~\ref{thm:generalized_15_puzzle}, we only need to argue that if $\mathcal{I}$ contains two or more tokens of a common color $c^*>1$ then it is a yes-instance.

    We define two color-preserving bijections $\mu', \mu''$ between tokens of the initial and final configuration of $\mathcal{I}$.
    The first bijection $\mu'$ can be arbitrarily (and its existence is guaranteed since the number of tokens of the same color is the same in both the initial and the final configuration).
    Let $T$ be the set of tokens of color other than $1$ in the initial configuration.
    The second bijection $\mu''$ is defined by $\mu''(t) = \mu'(t)$ for all $t \in T \setminus \{t_1, t_2\}$, $\mu''(t_1) = \mu'(t_2))$, and $\mu''(t_2) = \mu'(t_1)$. 
    
    The bijection $\mu'$ (resp.\ $\mu''$) yields an instance $\mathcal{I}'$ (resp.\ $\mathcal{I}''$) of \cctss in which all tokens of color other than $1$ have distinct colors.

    Since the permutation induced by the initial and final positions of the tokens in $\mathcal{I}'$  and  $\mathcal{I}''$ have different parity (as they differ by a single transposition), at least one such instance $\mathcal{I}^* \in \{ \mathcal{I}', \mathcal{I}'' \}$ does not satisfy condition (iii) of Theorem~\ref{thm:generalized_15_puzzle}, and hence it is solvable.

    The claim follows by observing that, since $\mu'$ and $\mu''$ are color-preserving, any sequence solving $\mathcal{I}^*$ also solves $\mathcal{I}$.
\end{proof}

Since the instances of \cctss in which the base graph is a cycle or $T_0$ can be solved in polynomial time,\footnote{When the base graph is $T_0$, the instance can be solved by brute-force. When the base graph is a cycle, a final configuration is reachable if and only it preserves the cyclic sequence of token colors other than $1$.} the characterization of Lemma~\ref{lemma:ccts_p} implies the following.

\begin{theorem}
    The \cctss problem can be solved in polynomial time.
\end{theorem}

\section{Conclusions}
\label{sec:conclusion}
In this paper, we closed the problem left open in \cite{YangZ25} by showing that \ccts undergoes a sharp change in complexity when the number of colors increases from $3$ to $4$ (assuming $\mathsf{P} \neq \pspace$). In particular we showed that the problem is \pspace-hard when the swap graph contains an induced path on $4$ vertices.
This leaves open the problem of characterizing the computational complexity of \ccts when the swap graph is $P_4$-free.
A promising first step in this direction could be studying complete bipartite swap graphs, for which we believe the problem to be solvable in polynomial time.

\bibliographystyle{plainurl}
\bibliography{bibliography.bib}

\end{document}